\documentclass{icmart}

\contact[rinat@illinois.edu]{Department of Mathematics MC-382, University of Illinois, Urbana, IL 61801 USA}

\newcommand{\Q}{\mathcal Q}

\newtheorem{theorem}{Theorem}[section]

\newtheorem{lemma}[theorem]{Lemma}

\newtheorem*{coro}{Corollary} %%%% for unnumbered statements

\theoremstyle{definition}
\newtheorem{definition}[theorem]{Definition}
\newtheorem{example}[theorem]{Example}
\newtheorem{remark}[theorem]{Remark}

\newcommand{\Hom}{\operatorname{Hom}}

\title{Fermionic spectra in integrable models}

\author{Rinat Kedem}
\begin{document}

\begin{abstract}
  This is a brief review of several algebraic constructions related to
  generalized fermionic spectra, of the type which appear in integrable quantum
  spin chains and integrable quantum field theories. We 
  discuss the connection between fermionic formulas for the graded
  dimensions of the spaces of conformal blocks of WZW theories,
  quantum cluster algebras, discrete integrable noncommutative
  evolutions and difference equations.

\end{abstract}



\maketitle

\section{Partition functions in statistical mechanics and conformal
  field theory}
In statistical mechanics, a fundamental object of interest is the
partition function, the sum over the space of configurations
$\mathcal C$ of the Boltzmann
weight $e^{-E/kT}$, where $E$ is the energy of a configuration:
$$
Z = \sum_{c\in \mathcal C} e^{-E(c)/kT}.
$$

If the lattice is two-dimensional, the standard test for integrability is
the existence of a commuting family of transfer matrices. For a system
with periodic boundary conditions, the partition function can be
written as the trace of the product of transfer matrices. These are
operators on the Hilbert space of a one-dimensional slice of the lattice,
which depend on a spectral parameter. The coefficients in expansion of
this operator as a series in the spectral parameter gives commuting
integrals of the motion, hence integrability.

The Hamiltonian associated with the one-dimensional system is one of
those integrals. For example, the six-vertex
model transfer matrix is associated with the $XXZ$ Heisenberg spin
chain Hamiltonian \cite{baxter}. 

The two-dimensional integrable lattice model may undergo a
second-order phase transition at certain critical points, in the
infinite-lattice limit.  At the critical point, the behavior of the
model may be described by an effective conformal field theory.  The
correspondence includes the identification of the critical exponents,
given by the conformal dimensions, and the specific heat, given by the
central charge of the family of Virasoro representations which make up
the Hilbert space of the quantum field theory. It was shown in
\cite{KM,KKMM} that the massless part of the spectrum -- that is,
order 1 excitations in the statistical model, and quasi-particles in
the quantum field theory -- are also related, and an identification
can be made via the partition functions.

For the lattice model the spectrum is computed from the Bethe ansatz.
The ``order one'', or massless, excitations, which contribute to the
conformal partition function, have a ``quasi-particle-like''
behavior. For small momenta, their energy is a linear function of the
momentum. We call this the linearized spectrum.

In conformal field theory, the chiral part of the partition function
is given by the specialized characters of certain (not necessarily
irreducible) Virasoro modules. The full partition function is a
modular invariant sesquilinear combination of these and includes both the
chiral and anti-chiral parts.

In the original work of the author and collaborators \cite{KM,KKMM},
it was shown that, starting from the Bethe ansatz,
linearizing the spectrum and considering only massless
excitations above the vacuum, the resulting partition
function is equal to the chiral part of the partition function in the
conformal field theory, given by Virasoro characters.

The spectrum obtained from the Bethe ansatz is invariably of 
fermionic nature. At the time when this work was done, few fermionic
constructions of Virasoro modules were known. For example, the
Feigin-Fuks construction of the most interesting Virasoro modules
involves a resolution of the Verma module using the singular vectors,
and is thus given by an inclusion-exclusion principle, or (in general) an
infinite alternating sum.

A fermionic construction is a basis of the representation given by the
action of skew-commuting operators on the vacuum. This gives rise to
fermionic statistics: Identical fermions cannot occupy the same point
in phase space. One type of generalization of fermionic statistics
will be given below. These rules are combinatorial and this is
reflected in the expression for the partition function.

There are various ways of constructing bases for any given Virasoro
module. The idea of fermionic constructions is that physically
meaningful ones reflect the spectrum away from criticality of the
integrable quantum field theory. The particle content is some
reflection of the form of the primary fields of the conformal field
theory. These fields are the generalized fermions. See \cite{JMS} for
a recent example of this.

This note is organized as follows. In Section 2 we will give a few
examples of fermionic partition functions related to WZW models. In
Section 3, we will relate the general fermionic formulas for graded
dimensions of the space of conformal blocks to cluster algebras and
quantum cluster algebras. In Section 4, we will show how the
integrability of the resulting discrete difference equations
(Q-systems and their quantized version) can be used to give difference
equations satisfied by generating functions for the graded dimensions
of the space of conformal blocks. These are variants of quantum
difference Toda equations. These dimensions are the dimensions of the
moduli space of holomorphic vector bundles on the sphere with
prescribed punctures, and their graded analogs.

\subsection{Acknowledgements} 
The author would like to express her gratitude to her advisor
B.M. McCoy, with whom the original formulation of the physical
interpretation of conformal partition functions and their fermionic
expressions was made; To M. Jimbo, T. Miwa for their patient mentorship
over many years; to them as well as B. Feigin and E. Ardonne,
with whom she first worked on fusion products; and most especially to
P. Di Francesco for an ongoing illuminating collaboration. The author
thanks S. Fomin, H. Nakajima, N. Reshetikhin for their kindness and for
helpful discussions related to this work over the years.  This work
has been supported by the National Science foundation through several
grants, most recently NSF DMS grant 1100929.

\vskip.2in

\section{Generalized fermionic formulas}

Let us be specific about what we mean by a generalized fermion and the
resulting fermionic formula for the partition function. This
phenomenon occurs in finite or infinite systems. The natural
finite-dimensional system to start from is a solvable model on the
finite lattice with a spectrum governed by the statistics of the Bethe ansatz
equations. The eigenstates of the Hamiltonian with a Bethe ansatz
solution are in bijection with solutions of a coupled set of algebraic
equations. The solutions are specified by a set of integers chosen
distinctly on certain finite intervals. We interpret a choice of one
integer as a quasi-particle, and a choice of $m$ integers as $m$
quasi-particles. The corresponding choice of integers is proportional
to their momentum, one of the conserved quantities. The fact that the
integers should be distinct is what gives them a fermionic
nature.

 The resulting combinatorics is as follows. We approximate the energy
of each quasi-particle as a linear function of the momentum (they are
massless) and hence the Bethe integers. This is a reasonable
assumption in the conformal, infinite-size limit.
Suppose the Hilbert space with $m$ quasi-particles has $m$ integers
chosen distinctly from the set $[1,p+m]$ for some integer $p\geq 0$. 
Let $q=e^{-\alpha}$, where $\alpha$ is the proportionality constant
between the energy and the Bethe integers.
Then the partition function of $m$ quasi-particles is
$$q^{m(m+1)/2}
{\left[ \begin{array}{c} p+m \\ m \end{array}\right]}_q $$
where the $q$-binomial coefficient is defined as
$$
\left[ \begin{array}{c} p +m \\ m \end{array}\right]_q =
\prod_{j=1}^m\frac{(1-q^{p+j})}{(1-q^j)}, \quad p\geq m,
$$
and is defined to be zero if $p< m$.
The partition function of fermions on the interval $[1,p+m]$ is
$$Z_p(q) = \sum_{m\geq 0} q^{m(m+1)/2} {\left[ \begin{array}{c} p+m \\
      m \end{array}\right]}_q.$$ 
In the limit $p\to \infty$, this formula becomes
$$
Z(q) = \sum_{m\geq 0} q^{m(m+1)/2}\frac{1}{\prod_{j=1}^m (1-q^j)}.
$$
The important characteristics to note are
\begin{enumerate}
\item There is a quadratic function of the particle number $m$ in the
  exponent. This is the ``ground state energy'' of a fermionic system
  with $m$ particles.
\item There is a $q$-binomial coefficient, or its $p\to\infty$ limit,
  which is just the weighted sum over configurations above the ground
  state of $m$ fermions.
\end{enumerate}
A slight generalization of fermionic statistics always occurs in the
Bethe ansatz solution: The integer $p$ is a linear function of $m$
itself, in addition to the external parameters of the system (such as
size). 

Moreover, there is in general more than one ``color'' of
quasi-particle, and these have available energy ranges for each color
separately.  Again, these are free fermions, except for the
generalized statistic which hides in the integers $p_i$ for each
color: Each $p_i$ is a function of the number $m_j$ of 
quasi-particles of type $j$ in the system.

Thus a fermionic formula for the (conformal, linearized version of
the) partition on the finite lattice might has the form 
\begin{equation}\label{fermionpartition}
Z(q) = {\sum_{\mathbf m}}^{(1)} q^{Q(\mathbf m)} \prod_i
\left[ \begin{array}{c} p_i +m_i\\ m_i \end{array}\right]_q. 
\end{equation}
Here, $\mathbf m\in \mathbb Z_+^k$ for some $k$, The ground state
energy $Q({\mathbf m})$ is a quadratic function of the particle
content $\mathbf m$ which depends on the model, as are $p_i$, which
are in general linear functions of $\mathbf m$, and may tend to
infinity as the size of the system becomes infinite. Here, the
superscript $(1)$ on the summation indicates possible restrictions on
the summation variables corresponding to symmetry sectors of the
Hamiltonian. A finite system will have only a finite number of terms
in the summation. Moreover there may be several different symmetry
sectors of the Hamiltonian, in which case the partition function can be
projected to the different sectors separately.

If the model has a conformal limit (the size of the system is
infinite while the spectrum remains linearized, that is, the system
remains critical), the partition function -- properly normalized and
restricted -- tends in the limit to the graded character of some
Virasoro module. That is, $Z(q)$ is proportional to the trace of
$q^{L_0}$ over the module, where $L_0$ is the grading element of the
Virasoro algebra. A conformal field theory is built out of such
modules. This gives a direct connection between the spectrum of the
lattice model and the conformal field theory in certain cases.

\vskip.2in

\subsection{The Fock space as a limit of the reduced wedge product}

It is well known that the basic representation of the affine algebra
$\widehat{\mathfrak sl}_n$ can be realized as a quotient of the Fock
space of free fermions by a Heisenberg algebra \cite{KacRaina}. It
is possible to give a finite-dimensional version of this construction
\cite{Ke04}. As it is closely connected to the graded tensor product
construction introduced below in Section
\ref{sec-fusion}, we briefly summarize it. This finite-dimensional
fermionic space gives -- in the inductive limit -- the Frenkel
Kac construction of the level-1 modules.

Let $V=V(\omega_1)\simeq \mathbb C^n$ be the defining representation of
$\mathfrak g={\mathfrak {sl}}_n$, and
$V(z)=V\otimes \mathbb C[z]$ a representation on which $\mathfrak{g}^-:=\mathfrak{g}\otimes
\mathbb C[t^{-1}]\subset \widehat{\mathfrak{sl}}_n$ acts
as $x\otimes f(t^{-1}) v = f(z) x v$ with $x\in \mathfrak{g}$, $f(t)\in \mathbb C[t]$, and $v\in
V(z)$.

Consider the $N$-fold tensor product 
$$V_N(z_1,...,z_N)=V(z_1)\otimes \cdots \otimes V(z_N)\simeq V^{\otimes N}\otimes
\mathbb C[z_1,...,z_N]$$
on which $\mathfrak{g}^-$ acts by the
usual co-product:
$$
\Delta_{{\mathbf z}}(x\otimes f(t)) = \sum_{i=1}^N x_{(i)} f(z_i^{-1})
$$
where $x_{(i)}$ indicates $x$ acting on the $i$th factor in the tensor
product. Obviously, this action commutes with the diagonal action of
the symmetric group $S_N$, simultaneously permuting factors in the
tensor product and variables $z_i$.

It also commutes with the action of
the negative part of the Heisenberg algebra $\mathcal H_-$, acting on
the space by multiplication by symmetric polynomials in
$z_1,...,z_N$. (Operators of the form ${\rm id} \otimes t^{-n}, n>0$). Thus,
we have three commuting actions. We quotient by the action of the
Heisenberg, and project onto the alternating representation of $S_N$,
and the result is called the reduced wedge space. It is a finite
dimensional space described explicitly as follows.

The quotient by the Heisenberg action is the quotient of
$\mathbb C[z_1,...,z_N]$ by symmetric polynomials of positive degree $I_N$. That is,
$$V_N[{\mathbf z}]/{\rm Im} \mathcal H^- = 
V^{\otimes N}\otimes \mathbb C[z_1,...,z_N]/I_N := V^{\otimes N}\otimes R_N$$

The space $R_N$ is isomorphic to the cohomology ring of the Flag
variety and to the regular representation of $S_N$. In particular,
it is finite-dimensional. It is a graded by the homogeneous
degree in $z_i$ and the action of the symmetric group preserves the
graded components. Thus,
$$
R_N \simeq \underset{\lambda\vdash N}{\oplus}W_\lambda \otimes
M_{\lambda,N},
$$
where $W_\lambda$ are the irreducible representations of $S_N$ and
$M_{\lambda,N}$ is a graded multiplicity space.  We also have the
decomposition
$$
V^{\otimes N} \underset{\mathfrak{g}\times S_N}{\simeq} \underset{\nu\vdash N,
  l(\nu)\leq n}{\oplus}V(\overline{\nu})\boxtimes W_\nu
$$
where $\overline{\nu}$ is the partition $\nu$ stripped of its columns
of length $n$, and $V(\lambda)$ are irreducible finite-dimensional
representations of $\mathfrak{g}$.

Taking the tensor product with $R_N$ and projecting onto the
alternating representation with respect to the diagonal action of
$S_N$, we identify $\nu = \lambda^t$. Thus, the reduced wedge space is
isomorphic to 
$$
\mathcal F_{N}\simeq \oplus V(\overline{\lambda}) \otimes M_{\lambda^t,N},
$$
where $\lambda^t$ is the transpose of $\lambda$. The hilbert
polynomial of $M_{\lambda^t,N}$ is a Kostka polynomial. In the limit
as $N\to\infty$, the properly normalized coefficient of $V(\lambda)$
is a character of the $W$-algebra, which is the centralizer of $\mathfrak{g}$
acting on the level-1 module of $\widehat{\mathfrak{g}}$, and the character of
$F_N$ tends to the character of the basic representation of the affine
algebra. Thus the reduced wedge product is a truncation of this
space, a Demazure module. 

We will give fermionic formulas for the generalizations of this Kostka
polynomial below.

\subsection{The Hilbert space of the generalized Heisenberg model}
We now give a very general setting which gives rise to fermionic
partition functions. The wedge space in the previous section is a
special case of this construction.

The fermionic formula of the type \eqref{fermionpartition} appears in
particular in the generalized Heisenberg spin chain with periodic
boundary conditions. This is a quantum spin chain, whose Hamiltonian
is derived via the $R$-matrix which intertwines tensor products of
Yangian modules $Y(\mathfrak g)$.  The simplest case of this is known
as the XXX spin chain, which was the subject of Bethe's original ansatz
\cite{Bethe}.

To define this spin chain, choose a the following data:
\begin{enumerate}
\item Any finite-dimensional Yangian module $V_0$. This is known as
  the auxiliary space.
\item A sequence of $N$ Yangian modules $\{V_1,...,V_N\}$ of KR-type
  (see below).
\end{enumerate}
The choice of non-isomorphic representations $V_i, i>0$
is the anisotropy of the model. 

Let $R_{ij}: V_i\otimes V_j\mapsto V_j\otimes V_i$ be the intertwiner
of finite-dimensional representations, known as (a rational)
$R$-matrix. For
generic spectral parameters, the tensor product is irreducicble and $R$
is unique, up to scalar multiple. The transfer matrix of the
generalized anisotropic Heisenberg model with periodic boundary
conditions is the trace over $V_0$ of the matrix
$M=R_{0,1}R_{0,2}\cdots R_{0,N}$. The transfer matrix $T_{V_0}$ is an
operator on the space $V_1 \otimes V_2\otimes \cdots \otimes V_N$,
which is the Hilbert space of the spin chain, also known as the
quantum space.

Since the $R$-matrix satisfies the Yang-Baxter equation, it follows
easily that the transfer matrices corresponding to different auxiliary
spaces commute. Expanding the transfer matrix as a series in in the
spectral parameter of $V_0$, each of the coefficients in the expansion
-- an element in an algebra acting on the Hilbert space -- commutes
with the other coefficients.  These coefficients therefore 
form a family of commuting integrals of motion. The spin chain is a
quantum integrable system. The quantum spin chain Hamiltonian is one
of the integrals.

This model has a Bethe ansatz solution, at least when the modules
$\{V_i\}$ are of Kirillov-Reshetikhin (KR)-type \cite{KR}
\cite{KR,Chari}. Such modules
are parameterized by a highest weight with respect to the Cartan
subalgebra of $\mathfrak g\subset Y(\mathfrak{g})$ and a spectral parameter. The
highest weight of a KR-module is a multiple of one of the fundamental
weights of $\mathfrak{g}$.  

The eigenvectors and eigenvalues of the Hamiltonian are given by
solutions of the Bethe equations. Solutions are parameterized in terms
of sets of distinct integers in the same manner described above. The
linearized spectrum is proportional to the sum of these integers. 

For this particular model, there is an arbitrary number of
quasi-particle species or ``colors'' for each root of the Lie algebra
$\mathfrak{g}$, which obey generalized fermionic statistics.
The statistics depends only on the Cartan matrix and the
highest weights of $\{V_i\}$. We will write down this function
explicitly, as it is key to the rest of the paper (we restrict our
attention here to simply-laced $\mathfrak{g}$ here for simplicity; The other
cases are explained in \cite{HKOTY,AK,DFK}).

The Hilbert space of the model is $V_1\otimes \cdots\otimes V_N$, so its
dimension is $\prod_i |V_i|$. (The ordering of these representations does
not effect the spectrum.) The partition function of the linearized
spectrum gives a graded version of this dimension, and we will provide a
representation theoretical interpretation of this grading.

Let $\lambda_1,...,\lambda_N$ be the highest weights of $V_1,...,V_N$
respectively. Each $\lambda_i$ is a multiple of one of the fundamental
weights, and therefore the choice of highest weights
is parameterized by a  multi-partition 
\begin{equation}\label{bnu}
\boldsymbol
\nu=(\nu^{(1)},\ldots,\nu^{(r)}), \quad \nu^{(a)}\vdash n_a,
\end{equation}
where the non-negative integers $\mathbf n=(n_1,...,n_r)$ are defined by
$\sum_{i=1}^N \lambda_i = \sum_{a=1}^r n_a \omega_a$ and $r$ is the
rank of the algebra. 

Define a set of integers ${\mathbf m}=(m_1,...,m_r)$ as follows:
$$
C\mathbf m = \mathbf n - \boldsymbol \ell, \quad \ell_a = \langle
\alpha_a,\lambda\rangle.
$$
for any choice of  a dominant integral weight $\lambda$ such that
${\mathbf m}\in \mathbb Z_+^r$. 
The evaluation of  the character ${\rm ch}_{\mathbf z} V(\lambda)$ of
the irreducible $\mathfrak{g}$-module $V(\lambda)$ at $\mathbf
z=(1,...,1)$ is the dimension of $V(\lambda)$.

\begin{theorem}
The linearized partition function $Z_{\boldsymbol \nu}(q)$ is the evaluation at
$\mathbf z=(1,...,1)$ of
$$
M_{\boldsymbol \nu}(q;\mathbf z) = \sum_\lambda M_{\boldsymbol \nu,\lambda}(q) {\rm ch}_{\mathbf z}(V_\lambda),
$$
where
\begin{equation}\label{pf}
M_{\boldsymbol \nu,\lambda}(q) = {\sum_{\boldsymbol\mu\vdash{\mathbf m}}} q^{Q({\boldsymbol \mu})} 
\prod_{a=1}^r\prod_{j\geq 1}\left[ \begin{array}{cc} p^{(a)}_{j}
    + \mu^{(a)}_{j}-\mu^{(a)}_{j+1} \\ \mu^{(a)}_{j}-\mu^{(a)}_{j+1} \end{array}\right]_q.
\end{equation}
The sum extends over all multipartitions $\mu^{(a)}$ of $m_a$, and
\begin{itemize}
\item The integers $p_{a,j}$ are the sum over the first $j$ rows of
  the integer sequence $\pi^{(a)} = \nu^{(a)} - \sum_{b} C_{a,b}\mu^{(b)}$;
\item The quadratic function in the exponent is
$$
Q(\boldsymbol\mu)=
\frac12
\sum_{a,b=1}^r
\sum_{i\geq 1} \mu_i^{(a)}C_{a,b}\mu_i^{(b)}.
$$
\end{itemize}
\end{theorem}

\begin{theorem}[Combinatorial Kirillov-Reshetikhin conjecture, \cite{DFK}] 
The sets of Bethe ansatz integers correctly count the dimension of the
Hilbert space of the
anisotropic Heisenberg model.
\end{theorem}
That is, when evaluated at $q=1$, Equation \eqref{pf} gives
an expression for the dimension of the space of $\mathfrak{g}$-linear
homomorphisms from the tensor product of KR-modules to the irreducible
representation $V(\lambda)$. This was known as the Kirillov-Reshetikhin
conjecture \cite{HKOTY}.

\begin{remark}\label{nsum}
The sum in \eqref{pf} is known as the ``$M$-sum'' in the language of
\cite{HKOTY}. There is a similar sum called the ``$N$-sum'', where the
definition of the $q$-binomial coefficient is continued to values of
$p<0$ by
$$
\left[\begin{array}{c} p+m \\ m
\end{array}
\right]_q =
\frac{(q^{p+1};q)_\infty(q^{m+1};q)_\infty}{(q;q)_\infty(q^{p+m+1};q)_\infty},
\qquad
(a;q)_\infty = \prod_{i\geq 0} (1-a q^i).
$$
The fact that $N(q)=M(q)$ is highly non-trivial; it was first
conjectured by \cite{HKOTY}, who showed that the $N$-sum gave the
correct dimension of the tensor product. It was later proven in 
\cite{DFK,DFKfusion} and shown to be closely tied with the Laurent property
\cite{FZ} of the quantum cluster algebra \cite{BZ} associated with 
the $Q$-system, defined below. 
\end{remark}

\begin{remark}
The sum \eqref{pf} is a generating function for certain Betti numbers
of quiver varieties in special cases \cite{Lusztig}, see also more recent
work giving a geometric context \cite{KodNaoi}.
\end{remark}

%\begin{remark}
%In the context of the completeness of the Bethe ansatz solutions of
%the Heisenberg model, the  
%\end{remark}

\subsection{Space of conformal blocks in WZW theory}
The formula for the linearized partition function of the Heisenberg
spin chain is of interest for several reasons. 

First, it is
known that, in special stabilized infinite limits, its conformal limit
is the Wess-Zumino-Witten model at a level which depends on the
representations $V_i$.

\begin{example} Let $\mathfrak{g}=\mathfrak{sl}_2$, set $V_i=
  V(k\omega_1)$ for all $i=1,...,N$, and consider the limit as the
  number of representations, $N= 2M$ becomes infinite.  Then the limit
  $M\to\infty$ of the normalized partition function
  $\lim_{M\to\infty}\widetilde{Z}_{2 M}(q;\mathbf z)$ is the
  character of the level-$k$ module of the affine Lie algebra
  $\widehat{\mathfrak sl}_2$ with highest weight $k \Lambda_0$.
\end{example}

\begin{example} Let $\mathfrak{g}=\mathfrak{sl}_n$ and $V_i=V(\omega_1)\simeq \mathbb C^n$. 
Then in the limit $N\to\infty$ the normalized, linearized partition function
\eqref{pf} is a Kostka polynomial \cite{Ke04}. In the conformal limit,
this gives a character of the $W_n$-algebra which centralizes
the action of $\mathfrak{g}$ when acting on the level-1 modules.
\end{example}

Another important role of the linearized partition function of the Heisenberg
model  is that it gives the dimension (at $q=1$) of
the space of conformal blocks of WZW theory (when $k\gg 1$ is an
integer). This is the dimension of the moduli space of holomorphic
vector bundles on a Riemann surface with $N$ punctures, with
specified monodromy given by the representations $V_i$, which are
taken to be arbitrary $\widehat{\mathfrak{g}}$-modules induced from KR-modules, localized
at distinct points. It is also known as the space of coinvariants.

\begin{remark}
  The reason we take $k$ to be integer is that the integrality
  property of the representations is used in the proof of the
  statement. The reason we require $k\gg 1$ is that for finite $k$,
  one has the Verlinde coefficients rather than the Littlewood
  Richardson coefficients
  for multiplicities of the irreducibles in the tensor product of 
  integrable modules affine algebra modules. We did not take this into
  account in \eqref{pf}.
 A separate conjecture for the fermionic formula of
  the linearized partition function of this space can be found in
  \cite{FL}. If $k$ is sufficiently large, the multiplicity is just
  as a sum of products of Littlewood Richardson coefficients or
  their generalization.
\end{remark}

We have a graded version of the dimension of the moduli space, meaning
we keep track of a certain grading or a refinement of the space. It is
known that, in special cases, this corresponds to keeping track of the
Betti numbers for a certain quiver variety, giving a geometric meaning
to the graded dimensions.

%\subsubsection{Relation to quiver representations}
%look at Feigin, Reineke paper to extract
\subsection{A grading on the tensor product}\label{sec-fusion}

It is known that the Hilbert space of the Heisenberg model, together
with the linearized spectrum of the Hamiltonian, in the limit when the
number of representations $V_i$ becomes infinite (taking all
$V_i\simeq V$, the defining representation, for example), gives the
characters of affine algebras in the limit as the (chiral) conformal
partition function. The relevant conformal field theory is the WZW
model at level 1. 

\begin{remark} There is also an explicit construction of this
  infinite-dimensional Hilbert space for the XXZ model, using the
  quantum affine algebra, using a stabilized semi-infinite tensor
  product \cite{DFJMN}. In this case it is possible to construct the
  transfer matrix in terms of intertwining operators which gives a
  direct connection with the deformed primary fields of the conformal
  field theory.
\end{remark}

Moreover, we identify the {\em dimension} of the Hilbert space of the
finite, inhomogeneous Heisenberg model with dimension of the space of
conformal blocks (for level $k$ sufficiently large).

These two facts form the motivation for the following definition of a
graded tensor product \cite{FL}. Whereas there are other definitions
of an ``energy function'' on the tensor product which defines a
grading on the tensor product in the case of quantum affine algebras
(these correspond to the XXZ model, or the limit $q\to 0$ in the case
of the crystal basis), the definition here refers only to the
undeformed current algebra.

\begin{remark}
  KR-modules are defined for three algebras: For the quantum affine
  algebra $U_q(\widehat{\mathfrak g})$, the Yangian $Y(\mathfrak{g})$, and the
  current algebra $\widehat{\mathfrak{g}}$. \cite{KR,Chari}. One of the
  consequences of the theorems of \cite{AK,DFK} is that these all have
  the same structure under restriction to the underlying finite
  dimensional algebra, $\mathfrak{g}$ or $U_q(\mathfrak{g})$ \cite{Ke11}. Here we use only
  the current algebra version.
\end{remark}

\begin{definition}
Let $V$ be a cyclic $\mathfrak{g}[t]=\mathfrak{g}\otimes \mathbb C[t]$-module, 
defined by the representation $\pi$.
We define the representation $\pi_\zeta$ on $V$ as follows. Given
$x\otimes f(t)\in \mathfrak{g}[t]$ and $w\in V$, $\pi_\zeta(x\otimes f(t)) w =
\pi(x\otimes f(t+\zeta)) w$, for some $\zeta\in \mathbb C^*$.
\end{definition}
That is, the localization takes place at $\zeta$. We use the shorthand
$V(\zeta)$ for the module with the action $\pi_\zeta$, even though the
vector space itself is simply $V$.

Now pick $V_i(\zeta_i)$ to be KR-modules of $\mathfrak{g}[t]$, with $1\leq i
\leq N$, with $\zeta_i\neq \zeta_j$ for all $i\neq j$. Let $v_i$ be
the cyclic, highest weight vector of $V_i$.
We have
$V_i(\zeta_i) = U(\mathfrak{g}[t]) v_i$, and the tensor product is also cyclic
(as long as the localization parameters are distinct):
$$
V_1(\zeta_1)\otimes \cdots \otimes V_N(\zeta_N) = U(\mathfrak{g}[t]) v_1\otimes
\cdots \otimes v_N.
$$
(The assumption that we have KR-modules is not essential at this point,
only that each of the modules $V_i$ is cyclic.)

The algebra $\mathfrak{g}\otimes \mathbb C[t]$ is graded by degree in $t$, and
so is its universal enveloping algebra. Let $U_i$ denote the graded
component. The action of $U_i$ on the tensor product of cyclic vectors
inherits this
filtration, and therefore we have a filtration of the tensor product
itself. The associated graded space of this filtered space is called
the Feigin-Loktev ``fusion'' product, $\mathcal F^*_{V_1,...,V_N}$.

\begin{theorem}[\cite{AK,DFK,Ke11}]\label{dim}
The associated graded space is
isomorphic to the tensor product of KR-modules as a $\mathfrak{g}$-module. That
is, it is independent of the localization parameters $\zeta_i$.
\end{theorem}
For the proof of this theorem, it is essential that $V_i$ are of
KR-type. The graded $\mathcal F^*_{V_1,...,V_N}$ is defined as a
quotient space, so in general, its dimension may be greater than the
dimension of the tensor product itself. It corresponds to the
``collision'' of all the points $\zeta_i$. 

The theorem about the dimension of this space was proven using a
function space realization for the space of conformal blocks
(coinvariants), and the use of the Kirillov-Reshetikhin conjecture
about the explicit fermionic formula for the dimension of this space
\cite{AK}. The final step in this proof uses a theorem of \cite{HKOTY}
and the proof of the ``$M=N$'' conjecture at $q=1$ in \cite{DFK}.

Let ${\boldsymbol \nu}$ be the parameterization of the collection of
the KR-modules as in \eqref{bnu}, and let
Let $\mathcal F^*_{\boldsymbol \nu,\lambda}[n]=\Hom_{\mathfrak{g}}(\mathcal
F_{\boldsymbol\nu}^*[n],V(\lambda))$. and consider the Hilbert polynomial
$\sum_{n\geq 0} q^n \dim \mathcal F^*_{\boldsymbol\nu,\lambda}[n]$.

The following strong version of Theorem \ref{dim} is proven in
\cite{DFKfusion}: 
\begin{theorem}\label{qdim}[\cite{DFKfusion}]
The Hilbert polynomial of the Feigin-Loktev graded tensor product is
equal to the conformal partition function \eqref{pf}.
\end{theorem}
We will introduce an expression for the partition function \eqref{pf}
as a constant term in the product of solutions of the
$Q$-system, a discrete recursion relation, in the next section. At the
same time, we will identify the $Q$-system as a mutation in a cluster
algebra, which therefore has a natural $q$-deformation. The proof of
Theorem \ref{qdim} will uses the methods of \cite{DFK} applied to this
quantum cluster algebra.

\begin{remark}In special cases, the FL graded tensor product is an
  affine Demazure 
  module \cite{FoLi}, which has a grading by the Cartan element $d$ of
  the affine algebra. This grading is essentially the same as the
  FL-grading. Therefore we are guaranteed that the appropriate
  semi-infinite graded tensor product is the full affine algebra module.
\end{remark}

By definition \cite{AK}, the idea of an associated graded space is
equivalent to taking all the spectral parameters $\zeta_i\to 0$. The
sum over the multipartitions $\mu$ in equation \eqref{pf} can be
viewed as a sum over all possible desingularizations of this
degeneracy (this is evident from the derivation using functional space
realization in \cite{AK}, see also \cite{FS}).

\section{Difference equations and the fermionic formulas}
It was originally observed in the context of the completeness
conjecture of the Bethe ansatz, and later by the original attempt at
proving the combinatorial Kirillov-Reshetikhin conjecture \cite{HKOTY},
the fermionic sum $M_{\boldsymbol \nu}(q;z)$ in Equation \eqref{pf} is
closely related to a difference equation called the $Q$-system.

Let $\chi_{a,k}$ be the character of the KR-module with highest weight
$k \omega_a$, restricted to ${\mathfrak{g}}\subset {\mathfrak{g}}[t]$.
\begin{example}
If ${\mathfrak{g}}=\mathfrak{sl}_n$ then the KR-modules are irreducible under the
restriction to ${\mathfrak{g}}$, and are the modules with ``rectangular highest
weights''. In that case, $\chi_{a,k}$ is a Schur function
$S_{(a)^k}(z_1,...,z_n)$ with $\prod z_i=1$.
\end{example}

For any Lie algebra, the functions $\chi_{a,k}$ satisfy a simple difference
equation: In the case where ${\mathfrak{g}}$ is simply-laced, this is a two-step
recursion relation. Consider the system
\begin{equation}\label{qsys}
Q^{(a)}_{k+1}Q^{(a)}_{k-1}= (Q^{(a)}_k)^2 - \prod_{b\neq a} (Q^{(b)}_k)^{-C_{ab}}.
\end{equation}
(The relation is only slightly more cumbersome for non-simply laced
algebras, and has a generalization for ${\mathfrak{g}}$ an affine algebra.)

A two-step recursion relation has a unique solution given initial
data. The natural initial data for the $Q$-system is
\begin{enumerate}
\item The character of the trivial representation is equal to 1,
  $\chi_{a,0}=1$. Therefore, set $Q^{(a)}_0=1$ for all $1\leq a \leq r$
  where $r$ is the rank of the algebra.
\item Identify $Q^{(a)}_1$ with the character of the fundamental
  KR-modules, $\chi_{a,1}$.
\end{enumerate}
\begin{theorem}[\cite{NakajimaKR}]\label{characters}
The characters of the Kirillov-Reshetikhin are solutions of the
$Q$-system \eqref{qsys} with the initial data (1) and (2).
\end{theorem}

The $Q$-system
is a specialization of the $T$-system, satisfied by the
transfer matrices of the XXZ spin chain, or by the $q$-characters
\cite{FrenResh} of the KR-modules.
\begin{remark}
  For any given quantum spin chain, one can derive Bethe ansatz
  equations from different functional relations, obtaining a different
  set of coupled algebraic equations (the Bethe equations). Although
  it is standard procedure to use Baxter's equation to derive Bethe
  equations, it is also possible to use the $T$-system, see e.g. \cite{DKM}. The
  resulting equations, their solution and linearized spectrum, take a
  different form depending on the original functional relation. This
  reflects the fact that in the degenerate case of a massless spectrum
  (the critical point) there may be several different descriptions of
  the spectrum as a quasi-particle spectrum. This degeneracy is
  resolved when a massive integrable perturbation is considered.
\end{remark}

The transfer matrices satisfy $T$-system relation, a conjecture
of Kirillov and Reshetikhin 
proved (for finite, simply-laced Lie algebras) by Nakajima
\cite{NakajimaKR}, using the realization of the representation theory
of the quantum affine algebra in terms of his quiver varities. The
$T$-system is satisfied by $q$-characters of the KR-modules
\cite{FrenResh}.  This is the algebraic Kirillov-Reshetikhin
conjecture. Nakajima even proved a deformed version of the
$T$-system which holds for twisted tensor products of
$KR$-modules. Theorem \ref{characters} follows from this work.

The relation of the $Q$-system to the multiplicity formulas starts as
follows.  Recall the definition of the ``$N$-sum'' in Remark
\ref{nsum}. Then there is a constant term identity for the $N$-sum in
terms of solutions of the $Q$-system. Define
\begin{equation}
Z_{\boldsymbol \nu,\lambda}(\mathbf Q_0,\mathbf Q_1)^{(k)} = \prod_{a=1}^r
Q_1^{(a)}(Q_0^{(a)})^{-1}\left( \prod_{i\geq 1} 
(Q_i^{(a)})^{\nu_i^{(a)}-\nu_{i+1}^{(a)}} \right) (Q_k^{(a)}
(Q_{k+1}^{(a)})^{-1})^{\langle \alpha_a,\lambda\rangle+1},
\end{equation}
where $Q_i^{(a)}$ are solutions of the $Q$-system \eqref{qsys}.
Define $\langle Z \rangle$ to be the constant term of $Z$ in
$\{Q_1^{(a)}\}_a$, evaluated at $\{Q_0^{(a)}=1\}_a$.
\begin{theorem}
Let 
$$N_{\boldsymbol\nu,\lambda}^{(k)}(1) = \langle
Z_{\boldsymbol\nu,\lambda}(\mathbf Q_0,\mathbf Q_1)^{(k)}\rangle. $$
Then there exists an integer $J$ such that whenever $k>J$,
$N_{\boldsymbol\nu,\lambda}^{(k)}(1)$ is independent of $k$, and is equal to
$N_{\boldsymbol\nu,\lambda}(1)$.
\end{theorem}
The proof is by induction, using direct computation, starting from the
fermionic formula \eqref{pf}.

We still need to show that $N=M$, however. Moreover, we need an
identity for the $q$-multiplicities themselves.  For this, we do not
need to use the representation theoretical interpretation of the
$Q$-system.  Instead, we will use the Laurent property of cluster
algebras.

\subsection{$Q$-systems as mutations in a cluster algebra} Here, we
give an interpretation of the variables $Q_k^{(a)}$ as cluster
variables in a cluster algebra.

A cluster algebra is the commutative algebra generated by the union of
{\em cluster variables}, defined recursively. It was
originally introduced by Fomin and Zelevinsky \cite{FZ} in a
representation theoretical context, but has been shown to have
applications far beyond the original motivation.  We refer to Fomin's
ICM lecture notes for a good overview \cite{Fomin}.

We use only the simplest version. Let $B$ be a skew symmetric $n\times
n$ integer matrix (equivalently, a quiver with no 1- or 2- cycles),
called the exchange matrix. Vertices of the quiver are numbered from
$1$ to $n$ and the integer $B_{ij}$ is the number of arrows from $j$
to $i$. Let $\mathbf x = (x_1,...,x_n)$ be formal (commutative) variables
associated with the vertices. Fix $1\leq j\leq n$ and define $x_j'$
\begin{equation}\label{exchange}
x_j' = \frac{\displaystyle\prod_{i:j\to i} x_i+\prod_{i:i\to j} x_i}{x_j}.
\end{equation}
This is called a mutation of $\mathbf x$ in the direction $j$, denoted by
the operation $\mu_j$. If $i\neq j$, $\mu_j(x_i)=x_i$. However, the
quiver itself changes under the mutation as follows:
\begin{itemize}
\item For any sequence $i\to j\to k$, add an arrow $i\to k$.
\item Reverse any arrows incident to $i$.
\item Erase any resulting 2-cycles.
\end{itemize}
The collection of generators of the cluster algebra is the result of
all possible sequences of mutations of $\mathbf x$. The pair $(B,\mathbf x)$ is
called the seed data.

Any $Q$-system (that is, Equation \eqref{qsys} and its
generalizations), can be shown to be a mutation in a cluster algebra
\cite{Ke08}. In the case of \eqref{qsys}, it is the cluster algebra
defined by the seed data $(\mathbf x_0,B)$ where
\begin{equation}\label{seed}
B=\left(\begin{array}{cc} 0 & -C \\ C & 0 \end{array} \right),\quad
\mathbf x_0=(Q_{0}^{(1)},\ldots,Q_{0}^{(r)};Q_{1}^{(1)},\ldots,Q_{1}^{(r)}).
\end{equation}
Note that we do not impose $Q_0^{(a)}=1$ at this stage.

The $Q$-system equations are a special subset of the mutations of the cluster
algebra\footnote{Although traditionally, the coefficients in a cluster
  algebra are taken to be $+1$, we keep the minus sign in the current
  context. This can be dealt with by (1) renormalizing the
  $Q$-variables or (2) introducing coefficients
  \cite{DFKcluster}. However this is irrelevant in the current
  context.}. They can be shown to be the equations a discrete
integrable system \cite{DFKqsys}, with the integrals of motion given
by those of the Toda system \cite{GSV}. This is due to the existence
of an integrable Poisson structure compatible with the cluster algebra
structure. Such a Poisson structure can always be deformed to give a
quantum system, which in the case of cluster algebras is called a
quantum cluster algebra \cite{BZ,FG} (see below).

Any cluster algebra (and a much larger class of discrete rational
evolution equations) can be shown to have a {\em Laurent
  property}. The transformation \eqref{exchange} is a rational
transformation. Although it is obvious after a single mutation, it is
not at all obvious after several steps of mutations that the rational
function is, in fact, a Laurent polynomial in the seed data, because
the term in the denominator is itself a polynomial in the initial seed data.
\begin{theorem}[Laurent property \cite{FZlaurent}]
  Any cluster variable in a cluster algebra is a Laurent polynomial in
  the cluster variables of any other seed in the cluster algebra.
\end{theorem}

Taking the $Q$-system with the initial seed data consistent with the
character interpretation, this implies the
following:
\begin{theorem}[\cite{DFKcluster}]\label{poly}
Any cluster variable (not just solutions of the $Q$-system) in the
cluster algebra with seed data \eqref{seed} is a {\em polynomial} in
the variables $(Q_1^{(1)},\ldots,Q_1^{(r)})$ after evaluation at $Q_0^{(a)}=1$.
\end{theorem}
\begin{proof}
  This is a consequence of the Laurent phenomenon and the fact that
  the right hand side (the numerator in the exchange relation) of
  \eqref{qsys} vanishes at $k=0$. 

  We illustrate this for the case of ${\mathfrak{g}}=\mathfrak{sl}_2$. The
  generalization is clear. Let $x$ be a cluster variable in the
  cluster algebra. Then $x(Q_{-1};Q_{0}) = Q_0^{-m} \sum_{n\in\mathbb Z}
  p_n(Q_0) Q_1^n$, where $p_n$ is a polynomial. Performing the
  exchange $Q_{1}=N(Q_0)/ Q_{-1}$, we have $x=Q_0^{-m} \sum_{n\in \mathbb Z}
  p_n(Q_0) Q_{1}^{-n} N(Q_0)^{n}$. If $n<0$, since $N(Q_0)^{n}$ in the
  denominator is a polynomial, it must cancel with the term
  $p_{n}(Q_0)$ in the numerator, because the result must be a Laurent
  polynomial. That is, for any $n<0$, $p_n(Q_0)$ is divisible by
  $N(Q_0)$. (Up until this point, the argument holds for any
  bi-partite cluster algebra of any rank.)  Therefore, $p_n(1)=0$ for
  any $n<0$. Therefore, the cluster variable $x$ under this evaluation
  has only terms $Q_1^{n}$ with $n\geq 0$. The generalization to
  arbitrary rank relies on the identical argument: All variables and
  indices should be changed to multi-variables and multi-indices.
\end{proof}

When applied to the $Q$-system, the theorem implies that the
KR-modules are generated as Groethendieck ring by the fundamental
KR-modules.

The importance of polynomiality is in the proof of the ``$M=N$
conjecture'' \cite{HKOTY} which is the final step in the proof of the
combinatorial KR-conjecture \cite{DFK,Ke11} and hence the
Feigin-Loktev conjectures \cite{AK}.
\begin{theorem}[\cite{DFK}]\label{nism}
The constant term in $Q_1^{(a)}$ of $Z_{\boldsymbol\nu,\lambda}(\mathbf Q_0=1,\mathbf Q_1)$
has no contributions from terms in the summation in which any of the integers
$p_{a,i}<0$. That is, $M_{\boldsymbol\nu,\lambda}(1)=N_{\boldsymbol\nu,\lambda}(1)$.
\end{theorem}

The Laurent phenomenon and the polynomiality theorem generalize to the
quantum $Q$-system.

\subsection{The quantum $Q$-system from quantum cluster
  algebras}\label{sec-quantumQ} 

We are interested in the $q$-graded
version of Theorem \ref{nism}. This is obtained by using a
$q$-deformation of the $Q$-system.
There is a constant term identity for the graded
partition function \eqref{pf} in terms of the solutions of the quantum
$Q$-system. Aside from enabling us to prove that ``$N(q)=M(q)$'', it
gives yet another interpretation of the grading of the
multiplicity.

Given any skew-symmetric exchange matrix, one can define a quantum
cluster algebra \cite{BZ,FG}, a deformation of the compatible Poisson
structure of the cluster algebra \cite{GSVpoisson}. A quantum cluster
algbra is a 
non-commutative algebra generated by the seed data obeying
$q$-commutation relations, together with all its mutations. The
combinatorial data is the same as in the classical case, and the
exchange matrix is still the same matrix $B$.

Performing this deformation for the cluster algebra of the $Q$-system,
one obtains a quantum $Q$-system:
\begin{equation}\label{quantumQ}
t^{\Lambda_{a,a}}\mathcal Q_{k+1}^{(a)}\mathcal Q_{k-1}^{(a)} = (\mathcal Q_k^{(a)})^2 -
\prod_{b\neq a} (\mathcal Q_k^{(a)})^{-C_{ab}},
\end{equation}
where $\mathcal Q_k^{(a)}$ generate a non-commutative algebra defined by
\eqref{quantumQ} and the commutation relations
\begin{equation}\label{commutation}
\mathcal Q_n^{(a)}\mathcal Q_{n+1}^{(b)}=t^{\Lambda_{ab}}\mathcal Q_{n+1}^{(b)}\mathcal Q_n^{(a)},
\end{equation}
with $\Lambda =|C| C^{-1}$. The variables
$\{\mathcal Q_n^{(1)},...,\mathcal Q_n^{(r)}\}$ commute.

We will eventually identify $q=t^{-|C|}$ in our derivation of the
$M$-sums below.

Given initial seed data $\mathbf x_0=(\mathcal Q_0^{(a)},\mathcal Q_1^{(a)})_{a\in[1,r]}$,
any cluster variable can be expressed as a Laurent polynomial in the
initial seed data, with coefficients in $\mathbb Z[t,t^{-1}]$ (the Laurent
phenomenon for quantum cluster algebras was proven in
\cite{BZ}). Therefore, any Laurent polynomial $M$ of cluster variables can
be expressed as a Laurent polynomial in terms of any initial cluster
seed, for example, $\{\mathcal Q_0^{(a)},\mathcal Q_1^{(a)}\}$. By using the
commutation relations \eqref{commutation}, this Laurent polynomial
can be written in a normal ordered form, as {a finite sum}
\begin{equation}\label{normal}
M = \sum_{\mathbf n,{\mathbf m}\in\mathbb Z^r}
\prod_{a=1}^r(\mathcal Q_0^{(a)})^{m_a}\prod_{b=1}^r(\mathcal Q_1^{(b)})^{n_b} f_{{\mathbf m},\mathbf n}(t)
\end{equation}
where $f_{\mathbf n,{\mathbf m}}(t)\in \mathbb Z[t,t^{-1}]$.

We define the analogue of ``a constant term'' identity in the quantum
case by taking the constant term of this expression in $Q_1^{(a)}$,
and by evaluating at $Q_0^{(a)}=1$. 

\begin{definition}
Given a Laurent polynomial $M$ in $\{\mathcal Q_0^{(a)},\mathcal Q_1^{(b)}\}_{a,b}$,
define its constant term evaluated at $\mathcal Q_0^{(a)}=1$ by first,
defining the coefficients $f_{{\mathbf m},\mathbf n}(t)\in \mathbb Z[t,t^{-1}]$ as in
\eqref{normal}, then defining
$$
\langle M \rangle = \sum_{{\mathbf m}} f_{{\mathbf m},0}(t).
$$
\end{definition}
(Note that it is important to perform the
evaluation {\em after} normal ordering the expression, otherwise, we
miss out on the $t$-grading.)
 
The quantum Laurent property can be shown to imply that for the
quantum $Q$-system, any cluster variable, after evaluation at
$\mathcal Q_0^{(a)}=1$, is in fact a polynomial in $\{\mathcal Q_1^{(a)}\}_a$ with
coefficients in $\mathbb Z[t,t^{-1}]$ (the analog of theorem \ref{poly}).

For a given finite sequence $\boldsymbol\nu$ and a fixed $k$, define
\begin{equation}\label{mprod}
M_{\boldsymbol \nu,\lambda}^{(k)} = \prod_{a=1}^r\left(\mathcal Q_1^{(a)}
(\mathcal Q_0^{(b)})^{-1}\right) \prod_{i\geq 1}^{\rightarrow} \prod_{a=1}^r
(\mathcal Q_i^{(a)})^{\nu_i^{(a)}-\nu_{i+1}^{(a)}} \prod_{a=1}^r
(\mathcal Q_k^{(a)}(\mathcal Q_{k+1}^{(a)})^{-1})^{\langle\omega_a,\lambda\rangle+1}
\end{equation}

Again, when $k$ is sufficiently large, $\langle M_{\boldsymbol
  \nu,\lambda}^{(k)} \rangle$ is independent of $k$.

Upon
multiplying by an appropriate power of $q$ and identifying the
deformation parameter $t$ of the cluster algebra as $q=t^{-|C|}$, we
have
\begin{theorem}[Constant term identity \cite{DFKfusion}]\label{constantterm}
$$
M_{\boldsymbol \nu,\lambda}(q^{-1}) = q^{h(\boldsymbol \nu,\lambda)} \langle  M_{\boldsymbol \nu,\lambda} \rangle.
$$
for $k$ sufficiently large. Here the normalization factor is
$$
h(\boldsymbol \nu,\lambda) = -\frac12 \sum_{a,b=1}^r\sum_{i\geq1}
\nu_i^{(a)}C^{-1}_{ab}\nu_i^{(b)} -\frac12 \sum_{a=1}^r C^{-1}_{aa}
\ell_a - \sum_{a,b=1}^r C^{-1}_{ab} \nu_1^{(b)}.
$$
\end{theorem}

The polynomiality property, which follows from the Laurent property for
the quantum $Q$-system, implies
\begin{lemma}[\cite{DFKfusion}]
The cluster variables in the quantum cluster algebra corresponding to
the $Q$-system, after normal ordering and evaluation at
$\Q_0^{(a)}=1$ for all $a$, are polynomials in $\{\Q_1^{(a)}\}_a$.
\end{lemma}
Thus, we have the graded version of the $M=N$ identity:
\begin{theorem}[\cite{DFKfusion}] 
In the summation in Equation \eqref{pf}, terms with $p_{a,i}<0$ do not
contribute to the sum in the $q$-graded version of the identity. That
is, $M_{\boldsymbol\nu,\lambda}(q)=N_{\boldsymbol\nu,\lambda}(q)$.
\end{theorem}

\section{Difference equations}
So far, we have said nothing about the integrability of the $Q$-system
and its $Q$-deformed version. But in fact, this is a two-step
recursion relation of rank $r$, and it has $r$ integrals of the
discrete evolution (which are in involution with each other with
respect to the Poisson structure of the cluster algebra, or the
commutation relations of the quantum cluster algebra). In
type $A$ for example, the solutions
$Q_{k}^{(a)}$ satisfy linear recursion relations with $r+2$ terms, and
with coefficients which are integrals of the motion (or constants). 

These integrals of the motion can be used to find
differential/difference equations satisfied by generating functions
for partition functions (characters of graded tensor products). This
derivation is analogous to the construction of the Whittaker
functions, which are solutions of the quantum Toda equations in the
case of classical Lie groups, where the integrals of the motion are
the Casimir elements of the algebra \cite{Kostant}. More recently
there has been a certain interest in the so-called Gaiotto vector,
which is the analog of the Whittaker vector for Virasoro algebras, or
some degenerate version thereof.

Since certain stabilized limits of the graded tensor products tend to
various Virasoro modules or integrable affine algebra modules, it is
useful to first write these equations for the finite tensor
product. The result are Toda-like equations satisfied by the
generating function (the relation of fermionic character formulas and
Toda equations was noted in, e.g. \cite{FeiginToda}). This can be used
to derive difference equations satisfied by the stabilized limits of
the graded tensor products, and even solve them in special cases. One
can obtain the character formulas of Feigin and Stoyanovskii, or of
spinon type, by analyzing these difference equations.

The analog of the Whittaker function in the case of the graded tensor
product is the generating function \cite{DFK14}
$$
G(q;\mathbf z,\mathbf y) = 
\sum_{\boldsymbol\nu ,\lambda} q^{f_1(\boldsymbol \nu)}
{\rm ch_{\mathbf z}}V(\lambda) 
M_{\boldsymbol \nu,\lambda}(q) \prod_{a,i}
(y_{a,i})^{\nu^{(a)}_i-\nu^{(a)}_{i+1}} ,
$$
with $f_1(\boldsymbol\nu)= \frac12 \sum_{a,b,i}
\nu_i^{(a)}C_{a,b}^{-1}\nu_i^{(b)} + \sum_{a,b} C^{-1}_{ab}\nu_1^{(b)}$.
Using the factorization formula of Section \ref{sec-quantumQ}, this has a
very simple form:
$$
G(q;\mathbf z,\mathbf y)=\langle \prod_{a=1}^r\Q_1^{(a)}(\Q_0^{(a)})^{-1}
\prod_{a=1}^r\prod_{j\geq 1}(1-y_j^{(a)}Q_j^{(a)})^{-1} \tau(\mathbf z)\rangle
$$
where 
$$
\tau(\mathbf z)=\lim_{k\to\infty}\sum_{\lambda}q^{f_2(\lambda)}
{\rm ch}_{\mathbf z}V(\lambda) \prod_{a=1}^r (Q_k^{(a)}(Q_{k+1}^{(a)})^{-1})^{\langle
  \alpha_a,\lambda\rangle + 1} ,
$$
where $f_2(\lambda)=-\frac12 \sum_a C^{-1}_{aa}\ell_a$.

We claim that $\tau(\mathbf z)$ plays the role of the Whittaker
vector, with the role of the Casimir elements played by the discrete integrals
of motion of the $Q$-system. These act by scalars on this function,
whereas they act as $q$-difference operators on the product
of $\Q$'s to the left. This is the origin of the difference equations
satisfied by the partition functions.

\section{Summary}
We reviewed here the role played by the fermionic formulas for the
 characters of graded tensor products of current algebra modules,
 their close connection with discrete integrable equations called
 $Q$-systems, and their $q$-deformations. In the process, we
 used the formulation of these systems in terms of (quantum) cluster
 algebras, and found that the grading coming from the affine
 algebra action on the tensor product can be identified with the
 grading coming from the $q$-deformation of the cluster algebra, hence
 from the natural Poisson structure satisfied by the cluster algebra
 variables. The resulting graded tensor products, in the stabilized,
 semi-infinite limit, give a constrution of affine algebra or Virasoro
 modules. The integrability of the quantum $Q$-system is closely
 connected with difference equations satisfied by the characters of
 these modules.

\vskip.2in

%\bibliography{refs}

\begin{thebibliography}{10}

\bibitem{AK}
Eddy Ardonne and Rinat Kedem.
\newblock Fusion products of {K}irillov-{R}eshetikhin modules and fermionic
  multiplicity formulas.
\newblock {\em J. Algebra}, 308(1):270--294, 2007.

\bibitem{baxter}
Rodney~J. Baxter.
\newblock {\em Exactly solved models in statistical mechanics}.
\newblock Academic Press, Inc. [Harcourt Brace Jovanovich, Publishers], London,
  1989.
\newblock Reprint of the 1982 original.

\bibitem{BZ}
Arkady Berenstein and Andrei Zelevinsky.
\newblock Quantum cluster algebras.
\newblock {\em Adv. Math.}, 195(2):405--455, 2005.

\bibitem{Bethe}
H.~Bethe.
\newblock Theorie der metalle i. eigenwerte und eigenfunktionen der linearen
  atomkette.
\newblock {\em Zeitschrift für Physik}, (71):2005--6, 1931.

\bibitem{Chari}
Vyjayanthi Chari.
\newblock On the fermionic formula and the {K}irillov-{R}eshetikhin conjecture.
\newblock {\em Internat. Math. Res. Notices}, (12):629--654, 2001.

\bibitem{DKM}
Srinandan Dasmahapatra, Rinat Kedem, and Barry~M. McCoy.
\newblock Spectrum and completeness of the three state superintegrable chiral
  potts model.
\newblock {\em Nucl.Phys.}, B396:506--540, 1993.

\bibitem{DFJMN}
Brian Davies, Omar Foda, Michio Jimbo, Tetsuji Miwa, and Atsushi Nakayashiki.
\newblock Diagonalization of the {$XXZ$} {H}amiltonian by vertex operators.
\newblock {\em Comm. Math. Phys.}, 151(1):89--153, 1993.

\bibitem{DFKcluster}
Philippe Di~Francesco and Rinat Kedem.
\newblock {$Q$}-systems as cluster algebras. {II}. {C}artan matrix of finite
  type and the polynomial property.
\newblock {\em Lett. Math. Phys.}, 89(3):183--216, 2009.

\bibitem{DFKqsys}
Philippe Di~Francesco and Rinat Kedem.
\newblock {$Q$}-system cluster algebras, paths and total positivity.
\newblock {\em SIGMA Symmetry Integrability Geom. Methods Appl.}, 6:Paper 014,
  36, 2010.

\bibitem{DFK14}
Philippe Di~Francesco and Rinat Kedem.
\newblock unpublished, 2014.

\bibitem{FL}
B.~Feigin and S.~Loktev.
\newblock On generalized {K}ostka polynomials and the quantum {V}erlinde rule.
\newblock In {\em Differential topology, infinite-dimensional {L}ie algebras,
  and applications}, volume 194 of {\em Amer. Math. Soc. Transl. Ser. 2}, pages
  61--79. Amer. Math. Soc., Providence, RI, 1999.

\bibitem{FeiginToda}
Boris Feigin, Evgeny Feigin, Michio Jimbo, Tetsuji Miwa, and Evgeny Mukhin.
\newblock Fermionic formulas for eigenfunctions of the difference {T}oda
  {H}amiltonian.
\newblock {\em Lett. Math. Phys.}, 88(1-3):39--77, 2009.

\bibitem{FG}
V.~V. Fock and A.~B. Goncharov.
\newblock Cluster {$X$}-varieties, amalgamation, and {P}oisson-{L}ie groups.
\newblock In {\em Algebraic geometry and number theory}, volume 253 of {\em
  Progr. Math.}, pages 27--68. Birkh\"auser Boston, Boston, MA, 2006.

\bibitem{Fomin}
Sergey Fomin.
\newblock Total positivity and cluster algebras.
\newblock In {\em Proceedings of the {I}nternational {C}ongress of
  {M}athematicians. {V}olume {II}}, pages 125--145. Hindustan Book Agency, New
  Delhi, 2010.

\bibitem{FZ}
Sergey Fomin and Andrei Zelevinsky.
\newblock Cluster algebras. {I}. {F}oundations.
\newblock {\em J. Amer. Math. Soc.}, 15(2):497--529 (electronic), 2002.

\bibitem{FZlaurent}
Sergey Fomin and Andrei Zelevinsky.
\newblock The {L}aurent phenomenon.
\newblock {\em Adv. in Appl. Math.}, 28(2):119--144, 2002.

\bibitem{FoLi}
G.~Fourier and P.~Littelmann.
\newblock Weyl modules, {D}emazure modules, {KR}-modules, crystals, fusion
  products and limit constructions.
\newblock {\em Adv. Math.}, 211(2):566--593, 2007.

\bibitem{DFK}
Philippe~Di Francesco and Rinat Kedem.
\newblock Proof of the combinatorial {K}irillov-{R}eshetikhin conjecture.
\newblock {\em Int. Math. Res. Not. IMRN}, (7):Art. ID rnn006, 57, 2008.

\bibitem{DFKfusion}
Philippe~Di Francesco and Rinat Kedem.
\newblock Quantum cluster algebras and fusion products.
\newblock {\em Int. Math. Res. Not. IMRN}, (doi: 10.1093/imrn/rnt004), 2013.

\bibitem{FrenResh}
Edward Frenkel and Nicolai Reshetikhin.
\newblock The {$q$}-characters of representations of quantum affine algebras
  and deformations of w-algebras.
\newblock In {\em Recent developments in quantum affine algebras and related
  topics ({R}aleigh, {NC}, 1998)}, volume 248 of {\em Contemp. Math.}, pages
  163--205. Amer. Math. Soc., Providence, RI, 1999.

\bibitem{GSVpoisson}
Michael Gekhtman, Michael Shapiro, and Alek Vainshtein.
\newblock Cluster algebras and {P}oisson geometry.
\newblock {\em Mosc. Math. J.}, 3(3):899--934, 1199, 2003.
\newblock \{Dedicated to Vladimir Igorevich Arnold on the occasion of his 65th
  birthday\}.

\bibitem{GSV}
Michael Gekhtman, Michael Shapiro, and Alek Vainshtein.
\newblock {\em Cluster algebras and {P}oisson geometry}, volume 167 of {\em
  Mathematical Surveys and Monographs}.
\newblock American Mathematical Society, Providence, RI, 2010.

\bibitem{HKOTY}
G.~Hatayama, A.~Kuniba, M.~Okado, T.~Takagi, and Y.~Yamada.
\newblock Remarks on fermionic formula.
\newblock In {\em Recent developments in quantum affine algebras and related
  topics ({R}aleigh, {NC}, 1998)}, volume 248 of {\em Contemp. Math.}, pages
  243--291. Amer. Math. Soc., Providence, RI, 1999.

\bibitem{JMS}
Michio Jimbo, Tetsuji Miwa, and Feodor Smirnov.
\newblock Fermions acting on quasi-local operators in the {XXZ} model.
\newblock In {\em Symmetries, integrable systems and representations},
  volume~40 of {\em Springer Proc. Math. Stat.}, pages 243--261. Springer,
  Heidelberg, 2013.

\bibitem{KacRaina}
V.~G. Kac and A.~K. Raina.
\newblock {\em Bombay lectures on highest weight representations of
  infinite-dimensional {L}ie algebras}, volume~2 of {\em Advanced Series in
  Mathematical Physics}.
\newblock World Scientific Publishing Co., Inc., Teaneck, NJ, 1987.

\bibitem{KKMM}
R.~Kedem, T.~R. Klassen, B.~M. McCoy, and E.~Melzer.
\newblock Fermionic sum representations for conformal field theory characters.
\newblock {\em Phys. Lett. B}, 307(1-2):68--76, 1993.

\bibitem{Ke04}
Rinat Kedem.
\newblock Fusion products, cohomology of {${\rm GL}_N$} flag manifolds, and
  {K}ostka polynomials.
\newblock {\em Int. Math. Res. Not.}, (25):1273--1298, 2004.

\bibitem{Ke08}
Rinat Kedem.
\newblock {$Q$}-systems as cluster algebras.
\newblock {\em J. Phys. A}, 41(19):194011, 14, 2008.

\bibitem{Ke11}
Rinat Kedem.
\newblock A pentagon of identities, graded tensor products, and the
  {K}irillov-{R}eshetikhin conjecture.
\newblock In {\em New trends in quantum integrable systems}, pages 173--193.
  World Sci. Publ., Hackensack, NJ, 2011.

\bibitem{KM}
Rinat Kedem and Barry~M. McCoy.
\newblock Construction of modular branching functions from {B}ethe's equations
  in the {$3$}-state {P}otts chain.
\newblock {\em J. Statist. Phys.}, 71(5-6):865--901, 1993.

\bibitem{KR}
A.~N. Kirillov and N.~Yu. Reshetikhin.
\newblock Formulas for the multiplicities of the occurrence of irreducible
  components in the tensor product of representations of simple {L}ie algebras.
\newblock {\em Zap. Nauchn. Sem. S.-Peterburg. Otdel. Mat. Inst. Steklov.
  (POMI)}, 205(Differentsialnaya Geom. Gruppy Li i Mekh. 13):30--37, 179, 1993.

\bibitem{KodNaoi}
Ryosuke Kodera and Katsuyuki Naoi.
\newblock Loewy series of {W}eyl modules and the {P}oincar\'e polynomials of
  quiver varieties.
\newblock {\em Publ. Res. Inst. Math. Sci.}, 48(3):477--500, 2012.

\bibitem{Kostant}
Bertram Kostant.
\newblock On {W}hittaker vectors and representation theory.
\newblock {\em Invent. Math.}, 48(2):101--184, 1978.

\bibitem{Lusztig}
G.~Lusztig.
\newblock Fermionic form and betti numbers. {\tt ArXiv:math/0005010}.

\bibitem{NakajimaKR}
Hiraku Nakajima.
\newblock {$t$}-analogs of {$q$}-characters of {K}irillov-{R}eshetikhin modules
  of quantum affine algebras.
\newblock {\em Represent. Theory}, 7:259--274 (electronic), 2003.

\bibitem{FS}
A.~V. Stoyanovskii and B.~L. Feigin.
\newblock Functional models of the representations of current algebras, and
  semi-infinite {S}chubert cells.
\newblock {\em Funktsional. Anal. i Prilozhen.}, 28(1):68--90, 96, 1994.

\end{thebibliography}
%\bibliographystyle{plain}

\section{References}
%\frenchspacing

\end{document}